\documentclass[11pt]{article}
\usepackage{amsmath,amsthm,amsfonts}

\usepackage[letterpaper,top=1in,bottom=1in,left=1in,right=1in]{geometry}

\usepackage{amsmath,amsthm,amscd,amsfonts,amssymb,amstexalg}
\usepackage{verbatim}
\usepackage[classfont=bold,langfont=roman,funcfont=typewriter,full]{complexity}

\usepackage{mathrsfs} %formal math script font

\DeclareMathOperator{\prob}{\mathbf{P}}

\DeclareMathOperator{\expt}{\mathbf{E}}

\newcommand{\zeroone}{0$-$1}

\newcommand{\reals}{\mathbb{R}}

\newcommand{\ints}{\mathbb{Z}}

\newcommand{\nats}{\mathbb{N}}

\theoremstyle{plain}
\newtheorem{theorem}{Theorem}[section]
\newtheorem{lemma}[theorem]{Lemma}

\newtheorem{corollary}[theorem]{Corollary}

\theoremstyle{definition}
\newtheorem{definition}{Definition}[section]

\newclass{\OPP}{OPP}
\newclass{\OP}{OP}
\newclass{\BPEXP}{BPEXP}
\newlang{\linspckt}{\textsc{Linear-size Series-Parallel Circuits}}
\newlang{\sat}{\textsc{sat}}
\newlang{\seth}{\mathbf{SETH}}
\newlang{\maxtwosat}{\textsc{Max-2-Sat}}
\newlang{\smt}{\textsc{smt}}
\newlang{\pit}{\textsc{pit}}
\newlang{\CAPP}{CAPP}
\newlang{\perm}{\textsc{Permanent}}
\newlang{\satisfiability}{\textsc{Satisfiability}}
\newlang{\uniquesat}{\textsc{unique-sat}}
\newcommand{\ksat}[1][k]{\lang{\textsc{$#1$-sat}}}

\newlang{\dkcsp}{\textsc{$(d,k)$-csp}}

\newlang{\cnf}{\textsc{cnf}}
\newlang{\maxsat}{\textsc{maxsat}}
\newlang{\dnf}{\textsc{dnf}}
\newcommand{\kcnf}[1][k]{\lang{\textsc{$#1$-cnf}}}

\newcommand{\maxksat}[1][k]{\textsc{max-$#1$-sat}}
\newlang{\cnfsat}{\textsc{cnfsat}}
\newlang{\cktsat}{\textsc{Circuit Sat}}
\newlang{\lincktsat}{\textsc{Linear-size Circuit Sat}}
\newlang{\turingsat}{\textsc{Turing Sat}}
\newlang{\ind}{\textsc{Independent Set}}
\newlang{\maxind}{\textsc{Max Independent Set}}
\newlang{\subiso}{\textsc{Subgraph Isomorphism}}
\newlang{\hamp}{\textsc{Hamiltonian Path}}
\newlang{\parityfunction}{\textsc{Parity}}
\newlang{\hamc}{\textsc{Hamiltonian Cycle}}
\newlang{\clique}{\textsc{Clique}}
\newlang{\colorability}{\textsc{Colorability}}
\newlang{\kcolorability}{\textsc{$k$-Colorability}}
\newlang{\threecolorability}{\textsc{$3$-Colorability}}
\newlang{\fourcolorability}{\textsc{$4$-Colorability}}
\newlang{\vertexcover}{\textsc{Vertex Cover}}
\newlang{\ksetcover}{\textsc{$k$-Set Cover}}
\title{A Satisfiability Algorithm for  Sparse Depth Two Threshold Circuits}
\author{{\large\sc Russell Impagliazzo}
\thanks{This research is supported by NSF
grant CCF-1213151 from the
Division  of Computing and Communication Foundations.
Any opinions,
findings and conclusions or
recommendations expressed in this material are those
of the authors and do
not necessarily reflect the
views of the National Science Foundation.}
\\
\and {\large\sc Ramamohan Paturi} \footnotemark[1] \\
\and {\large\sc Stefan Schneider} \footnotemark[1] \\
\\
Department of Computer Science and Engineering\\
University of California, San Diego\\La Jolla, CA 92093-0404, USA\\
E-Mail: \{russell, paturi, stschnei\}@cs.ucsd.edu
}
\date{April 2013}
\begin{document}

\begin{titlepage}
	\maketitle
	\begin{abstract}
  We give a nontrivial algorithm for the satisfiability problem for
  $cn$-wire threshold circuits of depth two which is better than
  exhaustive search by a factor $2^{sn}$ where $s= 1/c^{O(c^2)}$.  We
  believe that this is the first nontrivial satisfiability algorithm
  for $cn$-wire threshold circuits of depth two. The independently
  interesting problem of the feasibility of sparse $\zeroone$ integer
  linear programs is a special case. To our knowledge, our algorithm
  is the first to achieve constant savings even for the special case
  of Integer Linear Programming.  The key idea is to reduce the
  satisfiability problem to the \emph{Vector Domination Problem}, the
  problem of checking whether there are two vectors in a given
  collection of vectors such that one dominates the other
  component-wise.

  We also provide a satisfiability algorithm with constant savings
for depth two circuits with symmetric gates where the total weighted fan-in 
is at most $cn$.

  One of our motivations is proving strong lower bounds for $\TC^0$
  circuits, exploiting the connection (established by Williams)
  between satisfiability algorithms and lower bounds.  Our second
  motivation is to explore the connection between the expressive power
  of the circuits and the complexity of the corresponding circuit
  satisfiability problem.
\end{abstract}

	\thispagestyle{empty}
	\clearpage
\end{titlepage}
\section{Introduction}
Satisfiability testing is both a canonical $\NP$-complete problem
\cite{Coo71, Lev73} and one of the most successful general approaches
to solving real-world constraint satisfaction problems.  In
particular, optimized $\cnfsat$ heuristics are used to address a
variety of combinatorial search problems successfully in practice,
such as circuit and protocol design verification.  The exact
complexity of the satisfiability problem is also central to complexity
theory, as demonstrated by Williams \cite{Williams_2010_stoc}, who has
showed that any improvement (by even a superpolynomial factor compared
to exhaustive search) for the satisfiability problem for general
circuits implies circuit lower bounds.  Furthermore he has
successfully used the connection to prove superpolynomial size bounds
for $\ACC^0$ circuits using a novel nontrivial satisfiability
algorithm for $\ACC^0$ circuits, solving a long standing open problem
\cite{Williams_2011_ccc}.

This raises the questions: For which circuit models do nontrivial
satisfiability algorithms exist?  How does the amount of improvement
over exhaustive search relate to the expressive power of the model
(and hence to lower bounds)?  Can satisfiability heuristics for
stronger models than $\cnf$ be useful for real-world instances?

Both the connection to circuit lower bounds and to heuristic search
algorithms point to threshold circuits as the model to study next.
Bounded depth polynomial size threshold circuits $\TC^0$ are the next
natural circuit class stronger than $\ACC^0$. $\TC^0$ is a powerful
bounded depth computational model. It has been shown that basic
operations like addition, multiplication, division, and sorting can be
performed by bounded depth polynomial size threshold
circuits. \cite{ChandraStockmeyerVishkin_1984_siam,BoppanaSipser_1991}. In
contrast, unbounded fan-in bounded depth polynomial size circuits over
the standard basis (even when supplemented with mod $p$ gates for
prime $p$) cannot compute the majority function
\cite{BoppanaSipser_1991}.  However, our understanding of the
limitations of bounded depth threshold circuits is extremely weak.
Exponential lower bounds for such circuits are only known for the
special case of depth two and bounded weight
\cite{HajnalMaassPudlakSzegedyTuran_1987_focs}.  For larger depth
circuits, barely superlinear lower bounds are known on the number of
wires \cite{IPS97}.

On the other hand, satisfiability for depth two threshold circuits
contains as special cases some well known problems of both theoretical
and practical significance. $\cnfsat$ is one such special case, since
both conjunctions and disjunction are a special case of threshold
gates.  $\maxksat$, the optimization form of $\kcnf$ satisfiability,
is another special case, since the top threshold gate can count the
number of satisfied clauses for an assignment.  Even for
$\maxksat[3]$, no algorithms with a constant factor savings over
exhaustive search are known (although such an algorithm is provided
for $\maxksat[2]$ in \cite{Williams_2005_tcs}).  Another special case
is Integer Linear Programming (ILP), a problem that is very useful in
expressing optimization problems both in theory and practice.  Testing
the feasibility for a $\zeroone$ ILP is equivalent to testing the
satisfiability of a circuit with two levels, the bottom consisting of
threshold gates and the top level being a conjunction.  So both
theoretical and real-world motivation points us to trying to
understand the satisfiability problem for depth two threshold
circuits.

Santhanam \cite{Santhanam_2010_focs} gives an algorithm with constant
savings for linear size formulas of AND and OR gates with fan-in
two. However, this does not directly give an algorithm for depth two
threshold circuits, as converting a linear size threshold circuit into
a formula over AND and OR gates gives quadratic size.

In all of these related problems, a key distinction is between the
cases of {\em linear size} and {\em superlinear size} circuits.  In
particular, an algorithm with constant savings for depth two threshold
circuits of superlinear size would refute the \emph{Strong Exponential
  Time Hypothesis} ($\seth$) \cite{ImpagliazzoPaturi_1999_jcss}, since
$\kcnf$ for all $k$ can be reduced (via Sparsification Lemma
\cite{ImpagliazzoPaturiZane_1998_jcss}) to superlinear size depth two
threshold circuits \cite{Calabro_2009_phd}. 
($\seth$ says that for every $\delta <1$, there is a $k$ such that
$\ksat$ cannot be solved in time $O(2^{\delta n})$.)
However, for $\cnfsat$ and $\maxsat$, algorithms
with constant savings are known when the formula is {\em linear size}
\cite{Schuler_2005_jalg, DantsinWolpert_2006_sat,
  AustrinBenabbasChattopadhyayPitassi_2012}.  So, short of refuting
$\seth$, the best we could hope for is to extend such an improvement to
the linear size depth two threshold circuit satisfiability problem.

In this paper, we give the first improved algorithm, which obtains a
constant savings in the exponent over exhaustive search for the
satisfiability of $cn$-wire, depth two threshold circuits for every
constant $c$.  As a consequence, we also get a similar result for
linear-size ILP. Under $\seth$, this is qualitatively the best we could
hope for, but we expect that further work will improve our results
quantitatively. For example, our savings is exponentially small in
$c$, whereas in, e.g., the satisfiability algorithm of
\cite{ImpagliazzoMatthewsPaturi_2012_soda} for constant depth and-or
circuits, it is polylogarithmic in $c$.  We consider this just a first
step towards a real understanding of the satisfiability problem for
threshold circuits, and hope that future work will get improvements
both in depth and in savings.

While we do not obtain any new circuit lower bounds, there is some
chance that this line of work could eventually yield such bounds. For
example, if there is an algorithm for any constant depth threshold
circuit with super-inverse-polynomial savings in $c$, then $\NEXP
\not\in \TC^0$ by applying \cite{Williams_2010_stoc}.

Our main sub-routine is an algorithm for the {\em Vector Domination
  Problem}: given $n$ vectors in $\reals^d$, is there a pair of
vectors so that the first is larger than the second in every
coordinate?  We show that, when $d < c \log n$ for a constant $c$,
this problem can be solved in subquadratic time.  In contrast,
Williams \cite{Williams_2005_tcs} shows that solving even the Boolean
special case of vector domination with a subquadratic algorithm when
$d = \omega(\log n)$ would refute $\seth$.  We think the Vector
Domination Problem may be of independent interest, and might be used
to reason about the likely complexities of other geometric problems
within polynomial time.

\section{Notation}
Let $V$ be a set of variables with $|V| = n$. An \emph{assignment} on
$V$ is a function $V \to \{0,1\}$ that assigns every variable a
Boolean value. A \emph{restriction} is an assignment on a set $U
\subseteq V$. For an assignment $\alpha$ and a variable $x$,
$\alpha(x)$ denotes the value of $x$ under the assignment $\alpha$.

A $\emph{threshold gate}$ on $n$ variables $x_1, \ldots , x_n$ is
defined by $\emph{weights}$ $w_i \in \reals$ for $1 \leq i \leq n$ and
a \emph{threshold} $t$.  The output of the gate is $1$, if
$\sum_{i=1}^n w_i x_i \geq t$ and $0$ otherwise. The \emph{fan-in} of
the threshold gate is the number of nonzero weights. We call a
variable an input to a gate if the corresponding weight is nonzero.
We also extend the definition of a threshold gate to $d$-ary symmetric
gates whose inputs and outputs are $d$-ary.

For a collection of threshold gates, the \emph{number of wires} is the
sum of their fan-ins.  A \emph{depth two threshold circuit} consists
of a collection of $m$ threshold gates (called the {\em bottom-level
  gates}) on the same $n$ variables and a threshold gate (called the
{\em top-level gate}) on the outputs of the bottom-level gates plus
the variables. The output of the circuit is the output of the
top-level gate. We call a variable with nonzero weight at the
top-level gate a \emph{direct wire}. For a $d$-ary depth two threshold
circuit, the gates are $d$-ary gates and the top-level gate only
outputs Boolean values.  The number of wires of a depth two threshold
circuit is the number of wires of the bottom-level gates. We call a
threshold circuit \emph{sparse} if the the number of wires is linear
in the number of variables.

A \emph{satisfiability algorithm} for depth two threshold circuits is
an algorithm that takes as input a depth two threshold circuit and
outputs an assignment such that the circuit evaluates to 1 under the
assignment.

A \emph{linear function} on a variable set $x_1,\ldots, x_n$ is a
function $g(x_1,\ldots,x_n) = \sum_{i=1}^n w_i x_i$, where $w_i \in
\reals$ are called the \emph{coefficients}. The \emph{size} of a
linear function is the number of nonzero coefficients. A \emph{linear
  inequality} is an inequality of the form $g(x_1,\ldots, x_n) \geq
t$.

An algorithm for the \emph{Integer Programming Feasibility Problem}
takes as input a collection of linear inequalities on variables
$x_1,\ldots,x_n$ and outputs an assignment $\{x_1,\ldots,x_n\} \to
\ints$ such that all inequalities are satisfied. We call an inequality
of the form $0 \leq x_i \leq d-1$ a \emph{capacity constraint}.  In a
\emph{0-1 Integer Programming Feasibility Problem} each variables is
constrained to be 0 or 1.

We use $\tilde{O}(f(n))$ to denote the asymptotic growth of a function
$f$ ignoring polynomial factors. Informally, we say an algorithm is
\emph{nontrivial}, if its time is significantly better than exhaustive
search. If ${\cal A}$ is a satisfiability algorithm for circuits with
$n$ variables with run time $\tilde{O}\left(2^{(1-s)n} \right)$, we
call $s$ the \emph{savings} of the algorithm over exhaustive search.

For a vector $u$, we use $u_i$ for coordinate $i$. 

All logarithms are base $2$ unless noted otherwise.

\section{Results and Techniques}
The main contribution of the paper is a nontrivial satisfiability
algorithm for sparse threshold circuits of depth 2.  More precisely,
we prove the following:

\begin{theorem}
        \label{thm:Main}
        There is a satisfiability algorithm for depth two threshold
        circuits on $n$ variables with $cn$ wires that runs in time
        $\tilde{O}\left(2^{\left(1 - s\right) n}\right)$ where
        \begin{equation*}
                s = \frac1{c^{O(c^2)}}
        \end{equation*}
\end{theorem}

While the proof in Section \ref{Fanin} assumes a Boolean
inputs for simplicity, the proof easily extends to threshold circuits
with $d$-ary inputs, yielding the following corollary.

\begin{corollary}
  There is a satisfiability algorithm for depth two threshold circuits
  on $n$ $d$-ary variables with $cn$ wires that runs
in time
  $\tilde{O}\left(d^{\left(1 - s\right) n}\right)$ where
        \begin{equation*}
                s = \frac1{c^{O(c^2)}}
        \end{equation*}
\end{corollary}

In the following, we provide a high level description of our
algorithm.  Intuitively, there are two extreme cases for the bottom
layer of a linear size threshold
circuits of depth two.  

The first extreme case is when we have a linear number of gates each
with bounded fan-in $k$.  This case is almost equivalent to $\maxksat$
and can be handled in a way similar to
\cite{CalabroImpagliazzoPaturi_2009_iwpec,AustrinBenabbasChattopadhyayPitassi_2012}.
Consider the family of $k$-sets of variables given by the support of
each bottom-level gate.  A probabilistic argument shows that, for some
constant $c$, there exists a subset of about $n-n/(ck)$ variables $U$
so that at most one element from each of the $k$-sets in the family is
outside of $U$.  Then for any assignment to the variables in $U$, each
bottom-level gate becomes either constant or a single literal, and the
top-level gate becomes a threshold function of the remaining inputs.
To check if a threshold function is satisfiable, we set each variable
according to the sign of its weight.

The second extreme case is when we have a relatively small number of
bottom-level gates, say, at most $\epsilon n$, but some of them  might have
a large fan-in.  In this case, we could first reduce the problem
to $\zeroone$ ILP
by guessing the truth value of all bottom-level gates and the top gate, 
and then verifying
the consistency of our guesses.  Each of
our guesses are threshold functions of the
variables, so testing consistency of our guesses is equivalent 
to testing whether 
the feasible region of about 
$\varepsilon n$ linear inequalities has a Boolean solution.  
%Moreover, the
%ILP we get will have a much smaller number of constraints than
%the number of variables.

We then reduce such an ILP to the Vector Domination problem.  To
do this, we partition the variables arbitrarily into two equal size
sets.  For each assignment to the first set, we compute a vector 
where the $i$'th component corresponds to the weighted
sum contributed by the first set of variables to the $i$'th threshold gate.
For the second set of
variables, we do the same, but subtract the contribution from the
threshold for the gate.  It is easy to see that
the vectors corresponding to a  satisfying assignment
are a dominating pair.  Since
there are $N=O(2^{n/2})$ vectors in our set, and each vector is of
dimension $d = \epsilon n = 2 \epsilon \log N$, to get constant
savings, we need a Vector Domination algorithm that is subquadratic
when the dimension is much less than the logarithm
of the number of vectors.
The last step is to give such an algorithm, using a simple but
delicate divide-and-conquer strategy.

Finally, to put these pieces together, we need to reduce the arbitrary
case to a ``convex combination'' of the two extreme cases mentioned
above. To do this, we use  the Fan-In Separation Lemma  
which asserts that there must be a
relatively small value of $k$ so that there are relatively few gates
of fan-in bigger than $k$ but less than $ck$, for some constant
$c$. 
%This value of $k$ partitions the gates into the ``small fan-in
%part'' (fan-in less than $k$) and the ``big fan-in part'' (greater
%than $ck$).  
We show that, as in the first extreme case, for a random
subset $U$ of variables, the gates with fan-in less or equal to $k$
almost entirely
simplify to constants or literals after setting the variables in
$U$.  
Our selection of $k$ ensures that the number of gates of fan-in greater than $k$
is small relative to the number of remaining variables.
%Then even combining the ``big fan-in'', the remaining ``small
%fan-in'' part, and the small number of intermediate sized gates, the
%total number of gates left is small compared to the number of
%remaining variables.  
So we can apply the method outlined for the second extreme case.  
The Fan-In Separation Lemma 
is where our savings becomes
exponentially small.  Unfortunately, this lemma is essentially tight,
so a new method of handling this step would be needed to make the
savings polynomially small.

Since the Integer Programming Feasibility problem with capacity
constraints can be expressed as a depth two threshold circuit with an
AND gate as the top-level gate, the results translate directly to the
feasibility version of sparse integer programs with capacity
constraints.  We get
\begin{corollary}
  Let $\{g_1 \geq a_1,\ldots,g_m \geq a_m\}$ be a collection of linear
  inequalities in variables $x_1,\ldots,x_n$ with total size at most
  $cn$.  There is an algorithm that finds an integer solution to the
  linear inequalities with capacity constraints $0 \leq x_i \leq d-1$
  for all $i$ in time $\tilde{O}\left(d^{\left(1 - s\right) n}\right)$
  for
        \begin{equation*}
                s = \frac1{c^{O(c^2)}}
        \end{equation*}
\end{corollary}

The following two sections contain the details of the proof.  Section
\ref{Domination} introduces the \emph{Vector Domination problem} and,
for small dimension, gives an algorithm faster than the trivial
quadratic time. The feasibility of a $\zeroone$ ILP with a small
number of inequalities is then reduced to the Vector Domination problem, 
yielding
an algorithm for such $\zeroone$ ILP with constant savings.
A reduction from
depth two threshold circuits to $\zeroone$ ILP 
concludes that
section. In Section \ref{Fanin}, we show how to  reduce 
the $cn$-wire depth two threshold
circuits 
satisfiability problem 
to the special case with a small number of  bottom-level gates relative
to the number of variables.  
The remaining sections discuss generalizations of our result.

\section{Vector Domination Problem}
\label{Domination}
In this section we introduce the \emph{Vector Domination} problem and
give an algorithm faster than the trivial $O(n^2)$ for small
dimension.

\begin{definition}
  Given two sets of $d$-dimensional real vectors $A$ and $B$, the
  \emph{Vector Domination Problem} is the problem of finding two
  vectors $u \in A$ and $v\in B$ such that $u_i \geq v_i$ for all
  $1 \leq i \leq d$.
\end{definition}

\begin{lemma}
  Let $d \in \nats$ and $A, B \subseteq \reals^d$ with $|A| + |B| =
  n$. There is an algorithm for the Vector Domination problem that runs in time
\begin{equation*}
O\left( \binom{d + \log n + 2}{d+1} n\right) 
\end{equation*} 
\end{lemma}
\begin{proof}
  The claim is trivial for $n=1$ or $d=1$. In the latter case we can
  sort the set $A \cup B$ and then decide if such a pair exists in
  linear time.

  Otherwise, let $a$ be the median of the first coordinates of $A \cup
  B$. We partition the set $A$ into three sets $A^+$, $A^=$ and $A^-$,
  where $A^+$ contains all vectors $u \in A$ such that $u_1 > a$,
  $A^=$ contains all vectors such that $u_1 = a$ and $A^-$ contains
  all vectors such that $u_1 < a$. We further partition $B$ into set
  $B^+$, $B^=$ and $B^-$ in the same way. A vector $u \in A$ can only
  dominate a vector $v \in B$ in one of three cases:
  \begin{enumerate}
  \item $u \in A^+$ and $v\in B^+$
  \item $u \in A^-$ and $v\in B^-$
  \item $u \in A^= \cup A^+$ and $v \in B^= \cup B^-$ 
  \end{enumerate}

  For the first two cases we have $|A^+| + |B^+| \leq \frac{n}2$ and
  $|A^-| + |B^-| \leq \frac{n}2$ as we split at the median. For the
  third case, we know $u_1 \geq v_1$, hence we can recurse on vectors
  of dimension $d-1$. Since finding the median takes time $O(n)$ we
  get for the running time of $n$ vectors of dimension $d$
  \begin{equation*}
    T(n,d) = 2T\left(\frac{n}2,d\right) + T(n,d-1) + O(n)
  \end{equation*}

  To solve this recurrence relation, we want to count the number of
  nodes in the recurrence tree with $n'=\frac{n}{2^i}$ and
  $d'=d-j$. There are $\binom{i+j}{j}2^i$ possible paths from the root
  node to such a node, as in every step we either decrease $n$ or $d$,
  and there are $\binom{i+j}{j}$ possible combinations to do so, and
  if we decrease $n$ there are two possible children. Since computing
  the median of $\frac{n}{2^i}$ numbers takes time $O(\frac{n}{2^i})$
  the total time is upper bounded by
  \begin{align*}
    \sum_{\substack{0\leq i \leq \log n \\ 0 \leq j \leq d}}
    \binom{i+j}{j}2^i O\left(\frac{n}{2^i}\right) &= \sum_{0\leq i \leq \log n}
    \binom{i+d+1}{d} O(n) \\ &= \left(\binom{\log n +d+2}{d+1} -
      \binom{d+1}{d+1}\right) O(n) \\ &= \left(\binom{d + \log n
      + 2}{d + 1} - 1\right) O(n) 
  \end{align*}
\end{proof}

We can reduce $\zeroone$ ILP with few inequalities to the Vector
Domination Problem.

\begin{corollary}
  \label{lem:ilp}
  Consider a $\zeroone$ Integer Linear Program on $n$ variables and $\delta n$
  inequalities for some $\delta > 0$. Then we can find a solution in
  time
  \begin{equation*}
    2^{n/2} \binom{(1/2+\delta)n}{\delta n} \text{poly}(n)
    \leq 2^{(1/2 + \delta (\log(e) + \log\left(1 + \frac{1}{2\delta}\right)) n} \text{poly}(n)
  \end{equation*}

  Note that this algorithm is faster than $2^n$ for $\delta < 0.136$. 
\end{corollary}
\begin{proof}
  Separate the variable set into two sets $S_1$ and $S_2$ of equal
  size. We assign every assignment to the variables in $S_1$ and $S_2$
  a $\delta n$-dimensional vector where every dimension corresponds to
  an inequality. Let $\alpha$ be an assignment to $S_1$ and let
  $\sum_{i=1}^n w_{i,j}x_i \geq t_j$ be the $j$-th inequality for all
  $j$. Let $a \in \reals^{\delta n}$ be the vector with $a_j =
  \sum_{x_i\in S_1} w_{i,j} \alpha(x_i)$ and let $A$ be the set of
  $2^{n/2}$ such vectors. For an assignment $\beta$ to $S_2$, let $b$
  be the vector with $b_j = t_j - \sum_{x_i \in S_2}
  w_{i,j}x_i(\beta)$ and let $B$ be the set of all such vectors $b$.

  An assignment to all variables corresponds to an assignment to $S_1$
  and an assignment to $S_2$, and hence to a pair $a \in A$ and $b \in
  B$. The pair satisfies all inequalities if and only if $a$ dominates
  $b$. Since $|A|+|B| = 2^{n/2 +1}$ and the dimension is $\delta n$,
  we can solve the domination problem in time
  \begin{equation*}
    O\left( \binom{n/2 + \delta n + 3}{\delta n + 1} 2^{n/2+1}\right)
  \end{equation*}
\end{proof}

We now reduce the satisfiability of a depth two threshold circuit with
$\delta n$ bottom-level gates and any number of direct wires to the
union of $2^{\delta n}$ ILP problems.

\begin{corollary}
  \label{lem:algorithm}
  Consider a depth two threshold circuit on $n$ variables and $\delta n$
  bottom-level gates for some $\delta > 0$. We allow an arbitrary
  number of direct wires to the top-level gate. Then there is a
  satisfiability algorithm that runs in time
  \begin{equation*}
    2^{\delta n} 2^{n/2} \binom{(1/2+\delta)n}{\delta n} \text{poly}(n)
    \leq 2^{(1/2 + \delta (\log(e) + \log(1 + 1/2\delta)+1)) n} \text{poly}(n)
  \end{equation*}

  Note that this algorithm is faster than $2^n$ for $\delta < 0.099$. 
\end{corollary}
\begin{proof}
  For every subset $U$ of bottom-level gates, we solve the
  satisfiability problem under the condition that only the
  bottom-level gates of $U$ are satisfied. For an assignment to
  satisfy both the circuit and the condition that only gates in $U$
  are satisfied, it must satisfy the following system of inequalities:
  \begin{enumerate}
  \item For gates in $U$ with weights $w_1,\ldots,w_n$ and threshold
    $t$, we have $\sum_{i=1}^n w_i x_i \geq t$.
  \item For gates not in $U$ we require $\sum_{i=1}^n w_i x_i < t$,
    which is equivalent to $\sum_{i=1}^n -w_i x_i \geq -t
    + \min_i{w_i}$.
  \item Let $v_1,\ldots,v_n$ be the weights of the direct wires and
    let $s$ be the threshold of the top-level gate. Further let $w_U$
    be the sum of the weights of the gates in $U$. Then $\sum_{i=1}^n
    v_i x_i \geq s - w_U$.
  \end{enumerate}
  Note that this system contains $\delta n + 1$ inequalities, and the
  additional dimension adds only a polynomial factor to the time. 

  Since we need to solve a system of inequalities for every possible
  subset of bottom-level gates to be satisfied, we have an additional
  factor of $2^{\delta n}$, which gives the running time as claimed.
\end{proof}

Williams \cite{Williams_2005_tcs} 
introduced the  reductions used in  
Corollaries \ref{lem:ilp} and \ref{lem:algorithm}.
He considered a special case of the  
Vector Domination problem (called the Cooperative Subset Query problem)
where the entries in the vectors are 0 and 1 instead of 
arbitrary real numbers
. Applying the reduction from Corollary \ref{lem:ilp} to $\cnfsat$,
he concludes that an algorithm for solving the Cooperative Subset Query
problem with $d=\omega(\log n)$
that runs in time $O(f(d)n^\delta)$ for some $\delta < 2$ and a
time-constructible $f$ gives a $\cnfsat$ 
algorithm in time
$O(f(m)2^{(\delta/2) n})$ where $m$ is the number of clauses. Our
algorithm only works for $d < 0.136 \log n$, so it would be
interesting to see how far this can be pushed.

\section{Fan-In Separation}
\label{Fanin}
In this section we reduce the satisfiability of a depth two threshold
circuit with $cn$ wires to depth two threshold circuits with at most
$\delta n$ bottom-level gates by considering all possible assignments to
a random subset $U$ of
variables.  
The goal of the restriction is to eliminate all but a small fraction
of gates.
$U$ will consist of all but a fraction $O(1/(ck))$ of the variables
where $k$ is chosen such 
that there are only a small number of gates of fan-in larger than $k$
relative to the number of 
remaining variables.
Fan-In Separation Lemma shows how to find such a $k$.
 
\begin{lemma}[Fan-In Separation Lemma]
  Let $\mathcal{F}$ be a family of sets such that $\sum_{F \in
    \mathcal{F}} |F| \leq cn$. Further let $a> 1$ and $\epsilon>0$ be
  parameters. There is an $k \leq a^{c/\epsilon}$ such that
  \begin{equation*}
    \sum_{\substack{F \in \mathcal{F} \\ k < |F| \leq ka}}|F| \leq
    \epsilon n
  \end{equation*}
\end{lemma}
\begin{proof}
  Assume otherwise for the sake of contradiction. For $0 \leq i \leq
  \frac{c}{\epsilon}$, let $f_i$ be the sum of $|F|$ where $a^i <
  |F| \leq a^{i+1}$. By assumption we have $f_i > \epsilon n$ for all
  $i$. Hence $\sum_{i=0}^{c/\epsilon} f_i > cn$, which is a
  contradiction.
\end{proof}

\begin{lemma}
  Consider a depth two threshold circuit with $n$ variables and $cn$
  wires. Let $\delta >0 $ and let $U$ be a random set of variables such that
  each variable is in $U$ with some probability $p$
  independently. There exists a $p = \frac1{c^{O(c^2)}}$ such that
  the expected number of bottom-level gates that depend on at least
  two variables not in $U$ is at most $3\delta p n$.
\end{lemma}
\begin{proof}
  Let $\epsilon = \frac{\delta^2}c$ and $a = \frac{c^2}{\delta^2}$ and
  let $k$ be the smallest value such that there are at most $\epsilon
  n$ wires as inputs to gates with fan-in between $k$ and
  $ka$. Further let $p = \frac{\delta}{ck}$.

  Using the Fan-In Separation Lemma we get $k \leq
  \left(\frac{c^2}{\delta^2}\right)^{c^2/\delta^2}$.  We distinguish
  three types of bottom-level gates: Small gates, with fan-in at most
  $k$, medium gates with fan-in between $k$ and $ka$, and large gates
  with fan-in at least $ka$. For each type of gates,
we argue that the expected number of
  gates that depend on at least two variables not in $U$
  is bounded by $\delta p n$.

  For medium gates, the total number of wires is bounded by
  $\frac{\delta^2}c n$ and each gate contains at least $k$
  wires. Hence the number of medium gates is bounded by $\frac{\delta}{ck}
  \delta n = \delta p n$.

  Large gates contain at least $ka$ wires, hence the number of large
  gates is bounded by $\frac{c}{ka}n=\frac{\delta}{ck}\delta n=
  \delta p n$.

  For small gates, we argue as follows. 
 Let $m$ be the number of small gates and let
  $l_1,\ldots, l_m$ be their fan-ins. Let $X_i$ denote the event that
  gate $i$ depends on at least two variables not in $U$ and let $X$ be
  the number of such events. We have $P(X_i) \leq
  \binom{l_i}{2} p^2 \leq l_i^2p^2$ and therefore
  \begin{equation*}
    E[X] = \sum_{i=1}^m P(X_i) \leq \sum_{i=1}^m l_i^2p^2 \leq p^2 kcn
    = \delta p n
  \end{equation*}
\end{proof}

\begin{lemma}
\label{lemma:entire}
  There is a satisfiability algorithm for depth two threshold circuits
  with $cn$ wires that runs in time $2^{(1-s)n}$ for
  $E[s]=\frac1{c^{O(c^2)}}$.
\end{lemma}
\begin{proof}
  Let $\delta = \frac1{48}$ and $U$ as well as other parameters be as
  above. For every assignment to $U$,
  we have a depth two threshold
  circuit with $pn$ variables and $3\delta pn$ bottom-level gates in
  expectation. Since $3\delta = \frac1{16} < 0.099$, we can decide 
  the satisfiability of such
  a circuit using Corollary \ref{lem:algorithm} with constant
  savings. Let $s' = 1/2 - 3\delta (\log(e) + \log(1 + 1/6\delta)+1)
  \approx 0.15$ be the savings with our parameters.

  Let $T$ be the time for carrying out the entire procedure.  
  Since we are interested in the
  expected savings we consider the logarithm of the time and get
  \begin{equation*}
    E[\log(T)] = (1-p)n + (1 - s') pn = (1- s'p) n
  \end{equation*}
  and the lemma follows from $p = \frac1{c^{O(c^2)}}$. 
\end{proof}

Since $s$ is bounded by above by $1$, we can repeat the process a
constant number of times until we find a restriction such that the
savings is at least its expectation. This gives us our main result
Theorem \ref{thm:Main}.

\section{Generalization to Formulas}
\label{Generalizations}
In this section we discuss an extension of our main result to linear
size, constant depth \emph{threshold formulas}. A formula is a circuit
such that the output of every gate is an input to at most one other
gate.  A formula can be viewed as a tree where the internal nodes
correspond to gates and the leaves to bottom variables.  Note that a
circuit of depth two is always a formula.  The proof is a direct
generalization of our main proof.
\begin{corollary}
  There is a satisfiability algorithm for depth $d$ threshold formulas
  with $cn$ wires that runs in time $\tilde{O}\left(2^{\left(1 - s\right)
      n}\right)$ where
        \begin{equation*}
                s = \frac1{((d-1)c)^{O(((d-1)c)^2)}}
        \end{equation*}
\end{corollary}
\begin{proof}[Proof sketch]
  The algorithm chooses a random restriction such that at most $\delta
  n$ gates depend on more than one variable after restriction, where
  $\delta = \frac{1}{16}$ as before. As in the original proof, we
  take into account that there is only a single top-level gate, which
  does not simplify after restriction. The main difference to our main
  proof is the notion of the fan-in. Instead of considering the number
  of inputs to a gate, consider the \emph{size} of a gate. The size of
  a gate is the size of the subtree rooted at that gate. It is also an
  upper bound to the number of variables the gate depends on.

  For all $i \leq d$, the sum of sizes of all gates at depth $i$ is at
  most $cn$, since the circuit is a tree with at most $cn$
  wires. Hence the sum of sizes of all gates (minus the top-level
  gate) is at most $(d-1)c$. 

  Using the Fan-In Separation Lemma we can select a set $U$ of size
  $pn$ where $p = \frac{1}{((d-1)c)^{O(((d-1)c)^2)}}$ such that the
  number of gates that depend on at least two variables not in $U$ is
  at most $\delta n$. We can then write each remaining gate as a
  linear inequality, as each input is either a variable, a negated
  variable or a constant, which allows us to apply Corollary
  \ref{lem:ilp}.  
\end{proof}
%\appendix
\section{Generalization to Symmetric Gates}

In this section we discuss a second extension, to \emph{symmetric
  gates}. A gate is symmetric if the output depends only on the
weighted sums of the inputs. In particular, threshold gates are
symmetric. The proof of our main result does not directly generalize
to symmetric gates, but we give a different algorithm to decide the
satisfiability of depth two circuits consisting of symmetric gates
that follows similar ideas as our main proof. For this algorithm we do
however require that the weights are integer and small. Specifically,
we define the \emph{weighted fan-in} of a gate as the sum of the
absolute weights and the \emph{weighted number of wires} as the sum of
the fan-ins of all the gates. The result applies to circuits with a
weighted fan-in of $cn$.

The main difference between the two algorithms is the problem we
reduce it to after applying a restriction. In our main result, we
reduce the satisfiability of the simplified circuit to a (small)
system of linear inequalities. Here, we reduce to a system of linear
equations. We first give an algorithm for linear equations.

\begin{lemma}
\label{lem:equations}
  There is an algorithm to find a Boolean solution to a system of
  linear equations on variables $\{x_1,\ldots,x_n\}$ in time
  $\tilde{O}(2^{n/2})$.
\end{lemma} 
\begin{proof}
  We first reduce the problem to subset sum. Let $w_{i,j}$ be the
  weight of $x_i$ in the $j$-th equation and let $r_j$ be
  right-hand side of the $j$-th equation. Further let $D =
  \max_{i,j}\{w_{i,j},r_j\}$ be the largest such value. We define $s_i
  = \sum_{j} w_{i,j}D^j$ and $s=\sum_{j} r_{j}D^j$. Then there is a
  solution to the system of linear equations if and only if there is a
  subset of the $s_i$ that sums to $s$. 

  To solve the subset sum problem, partition the set of $s_i$ into two
  sets of equal size and list all $2^{n/2}$ possible subset sums
  each. We can then sort the lists in time $O(2^{n/2}n)$ and determine
  if there is a pair of numbers that sums to $s$. 
\end{proof}

We reduce the satisfiability problem of depth two threshold
circuits with small integer weights to a system of linear equations to
get the following result.

\begin{theorem}
  \label{lem:symmetric}
  There is a satisfiability algorithm for depth $2$ circuits with
  symmetric gates and weighted number of wires $cn$ that runs in time
  $\tilde{O}\left(2^{\left(1 - s\right) n}\right)$ where
        \begin{equation*}
                s = \frac1{c^{O(c^2)}}
        \end{equation*}
\end{theorem}

As before, we pick a random restriction with some parameter $p$, such
that most gates depend on at most one variable.

Given an assignment, we distinguish between the Boolean output of a
gate and the \emph{value}. The \emph{value} is defined as the weighted
sum of the inputs. Note that the value uniquely defines the output of
a symmetric gate. Unlike our main proof, we guess the value of the
remaining gates, including the top-level gate. Given a value for every
gate, we can write a system of linear equation. We then solve the
system of linear equations on $n$ Boolean variables in time
$\tilde{O}(2^{n/2})$ using Lemma \ref{lem:equations}.

We need the overhead for guessing the values to be smaller than the
savings achieved with solving the system of linear equations. For
this, it is crucial that both the number of remaining gates and the
number of values they can obtain is small. Here we use the requirement
that the weights are small. We defer the details of the calculation on
how many systems of linear equations we need to solve until section
\ref{sectionProof}

One possible approach would be to select $p$ using a fan-in separation
technique. However, we only achieved savings that are doubly
exponentially small in $c$ using this approach. To get better savings,
it is useful to model the interplay between the parameter $p$ and the
circuit as an explicit zero-sum game, where the first player's (the
algorithm designer) pure strategies are the values of $p$ and the
second player's (the circuit designer) pure strategies are the
circuits where all the gates have the same fan-in.  The payoff is the
difference between the saving of solving the subset sum problem and
the overhead of guessing the values of the gates.

The mixed strategies of the circuit designer are circuits of symmetric
gates with a weighted number of wires of at most $cn$, where each such
circuit is viewed as a distribution of the total number of wires among
gates of different weighted fan-in.  The mixed strategies of the
algorithm designer are distributions on the values of $p$.  We then
apply the Min-Max theorem to lower bound the expected value of the
game by exhibiting a distribution (with finite support) on the values
of $p$.  We search through the values in the support of the
distribution to find a $p$ that produces the expected value.  This
novel game-theoretic analysis yields an overall savings which is only
single exponentially small in $c$. Section \ref{sectionMinMax}
contains the details of the Min-Max approach.

\subsection{The Algorithm}
\label{sectionProof}

We develop the algorithm of Lemma \ref{lem:symmetric} in three
stages. In this section, we consider $p$ a parameter and present a
satisfiability algorithm for depth two circuits with symmetric gates
and weighted number of wires of $cn$. We further assume that all the
bottom-level gates have the same weighted fan-in $f$. The algorithm
achieves savings $s_{p,f}$ and for certain combinations of $p$ and $f$
the savings might be negative. In the second stage we extend the
algorithm to circuits with varying fan-in and show that the
savings of the algorithm is a convex combination of $s_{p,f}$. In the
last stage, in Section \ref{sectionMinMax} we show how to select a
$p$ such that the savings is at least $\frac1{c^{O(c^2)}}$ for any
distribution on $f$.

As we are mainly interested in the savings, we look at the logarithm
of the time complexity and bound its expectation.

\begin{lemma}
\label{lem:Savings}
Let $0 \leq p \leq 1$ be a parameter and $C$ be a depth two circuit
with symmetric gates, variables $V=\{x_1, \ldots, x_n\}$, a weighted
number of wires of $cn$, and weighted fan-in $f$ for all bottom-level
gates. There is an algorithm that decides the satisfiability for such
$C$ with time complexity $T$ such that $\expt[\log(T)] = (1-s_{p,f})n$
for
\begin{equation*}
                s_{p,f} = \left\{
                        \begin{array}{ll}
                                \frac{p}{4} & \text{ if } pf < \frac{1}{4c} \\
                                \frac{p}2 - \frac{c}f \log\left(8cpf\right) & \text{ otherwise}
                        \end{array} \right.
        \end{equation*}
\end{lemma}
\begin{proof}
  We select a random subset $U \subseteq V$ such that a variable is in
  $U$ with probability $(1-p)$ independently. We note that $\expt[|U|]
  = (1-p)n$.  For each of the $2^{|U|}$ assignments to $U$, we solve
  the satisfiability problem of the simplified circuit.  Bottom-level
  gates where all inputs are in $U$ are removed and the top-level gate
  is adjusted appropriately.  Gates that only depend on one input are
  replaced by a direct wire to the top-level gate with an appropriate
  weight and adjustment to the top-level gate.  For all gates with at
  least two remaining inputs, we guess the value of the gate and
  express the gate as a linear equation.  Similarly, we guess the
  value of the top-level gate to get another linear equation. We then
  solve the resulting system of linear equations on $n' = n - |U|$
  variables in time $\tilde{O}\left(2^{n'/2}\right)$ using Lemma
  \ref{lem:equations}.

  The critical part of the analysis is bounding the overhead from
  guessing the values of the gates.  We first bound the number of
  distinct values a gate can take.  The top-level gate can only take
  polynomially many different values.  Consider a bottom-level gate
  with fan-in $l \geq 2$ after applying an assignment to the variables
  in $U$.  We bounds the number of distinct values that the gate can
  take in two different ways.  The number of possible inputs, and
  hence the number of possible values is bounded by $2^l$.  On the
  other hand, since the value is an integer between $-l$ and $l$, the
  number of possible values for the gates is also upper bounded by
  $2l+1$.  Hence, we use $\min\{2^l,2l+1\}$ as an upper bound for the
  number of values of a bottom-level gate with fan-in $l$.

  Since we have a control on the number of distinct values taken by a
  gate by assumption, our overhead crucially depends on the number of
  exceptional gates, gates that depend on more than one variable after
  applying an assignment to the variables in $U$. Intuition says that
  the number of exceptional gates should be small.  If the fan-in of a
  gate is small, then we expect that it will simplified to depend on
  at most one variable after assigning values to the variables in $U$.
  On the other hand, there cannot be too many gates of large fan-in.
  While the intuition is simple, it is tricky to make it work for us
  in the general context. At this stage, our focus is on estimating
  the savings $s_{p,f}$ for the probability parameter $p$ and weighted
  fan-in $f$.

  Let $H$ be a random variable denoting the number of possible values
  the remaining gates can obtain.  Our estimation of $H$ and $s_{p,f}$
  involves two cases.  Let $t= \frac{1}{4c}$. We first consider the
  case $pf < t$.  Let $U' \subseteq V - U$ be the set of variables
  that appear in exceptional gates.  Our goal is to upper bound
  $\expt[\log(H)] \leq \expt[U']$.

  Consider a bottom-level gate. Let $X$ be the random variable
  denoting the number of its inputs not in $U$.  Let $f' \leq f$ be
  the number of variables the gate depends on, and let $X$ be the
  random variable denoting the number of its inputs not in $U$.  The
  distribution of $X$ is $\text{Bin}(f',p)$, hence we have $\expt[X] =
  f'p$.  Let the random variable $Y$ denote the number of variables
  that the gate can contribute to $U'$.  Since $U'$ is the set of
  variables appearing in exceptional gates, we have $Y=X$ for $X\geq
  2$ and $Y=0$ otherwise. Hence
\begin{align*}
  \expt[Y] &= \expt[X] - \prob[X=1] \leq f'p - f'p(1-p)^{(f'-1)} \\
  &\leq f'p(1 - (1-p)^{f'} ) \leq f'p(1-(1-f'p)) \\
  &= (f'p)^2 \leq (fp)^2
\end{align*}
by Bernoulli's inequality.  Hence, for any variable $x$ which is an
input to the gate, the probability $x$ belongs to $U'$ is at most
$\frac{\expt[Y]}{f} \leq p^2 f \leq \frac{p}{4c}$.  Since the total
number of wires is bounded by $c n$, we have $\expt[\log(H)] \leq
\expt[|U'|] = c n \frac{p}{4c} = \frac{p}{4} n$.

For the logarithm of the time complexity this yields
\begin{equation*}
  \expt[\log(T)] = \expt[|U|] +
  \expt\left[\frac12 \left(n-|U|\right)\right] +
  \frac{p}{4} n + O(\log n) \leq 
  n\left(1-\frac{p}{4}\right) + O(\log n)
\end{equation*}
where the logarithmic summand stems from guessing the value of the
top-level gate. We have $s_{p,f} = \frac{p}4$.

We now consider the case $pf \geq t$.  Suppose the $i$-th gate has
$l_i$ inputs that are not in $U$.  The expected value of $l_i$ is
$pf$.  There are at most $2l_i + 1$ possible values for the
gate. Since all the bottom-level gates have the same weighted fan-in
$f$, the number of bottom-level gates is at most $cn/f$ and
$\expt[\sum_{i=1}^{cn/f} l_i] = pcn$.  We bound the expected logarithm
of the number of possible values of all gates by
\begin{align*}
  \expt\left[\log\left(\prod_{i=1}^{cn/f}\right)
    \left(2l_i + 1\right)\right] &= (cn/f) \sum_{i=1}^{cn/f} \expt\left[\left(\log(2l_i + 1) f/cn\right)\right] \leq (cn/f) \log(2pf + 1) \\
  &\leq (cn/f) \log\left(8cpf\right)
\end{align*}
where we use the concavity of the logarithm function in the penultimate step
and the fact $pf \geq \frac{1}{4c}$ in the last step.

For the logarithm of the time complexity we get, 
\begin{equation*}
  \expt\left[|U|\right] + \expt\left[\frac12
    \left(n-|U|\right)\right] + cn/f
  \log\left(8cpf\right) + O(\log n) \leq
  n\left(1-\left(\frac{p}2 - \frac{c}{f} \log\left(8cpf\right)
    \right)\right) + O(\log n)
\end{equation*}
with savings $s_{p,f} = \frac{p}2 - \frac{c}f \log\left(8cpf\right)$.
\end{proof}

We now extend the algorithm to circuits with varying fan-in and show
that the logarithm of the time complexities is lower bounded by a
convex combination of the savings $s_{p,f}$.  We model the $cn$-wire
circuits of varying weighted fan-in by a distribution $\mathcal{F}$ on
wires.  For each weighted fan-in $f$, the wire distribution
$\mathcal{F}$ specifies the number $c_f n$ of wires of bottom-level
gates of weighted fan-in $f$.  We denote the savings of our algorithm
on circuits with wire distribution $\mathcal{F}$ by
$s_{p,\mathcal{F}}$.
\begin{lemma}
\label{lem:Convex}
Let $0 \leq p \leq 1$ be a parameter and $C$ be a depth two circuit with
symmetric gates, $n$ variables and a weighted number of wires of $cn$,
where the wires are distributed according to $\mathcal{F}$. There is a
satisfiability algorithm for such $C$ with time complexity $T$ such
that $\expt[\log(T)] = (1-s_{p,\mathcal{F}})n$ for
        \begin{equation*}
                s_{p,\mathcal{F}} \geq \sum_{f=1}^n \frac{c_f}c s_{p,f}
        \end{equation*}
\end{lemma}
\begin{proof}
  The algorithm is the same as above.  The logarithm of the overhead
  of guessing the values for all bottom-level gates with fan-in $f$ is
  $\log(H_f) = \frac{c_fn}f \log\left(8cpf\right)$ if $pf \geq t$ and
  $\log(H_f) = \frac{c_f}c \frac{p}{4} n$ otherwise.  Solving the
  system of linear equations and using linearity of expectation then
  yields the savings as claimed.
\end{proof}

\subsection{The Algorithm as a Zero-Sum Game}
\label{sectionMinMax}
The time complexity of the algorithm in Section \ref{sectionProof}
depends crucially on choosing a suitable parameter $p$. Instead of
trying to directly determine a good parameter $p$ by analyzing the
wire distribution of the circuit, we apply a trick from game theory.

A zero-sum game with two players $\A$ and $\C$ is a game where both
players pick a strategy and the outcome is determined by a function of
the two strategies.  Player $\A$ tries to maximize the outcome, while
player $\C$ tries to minimize it.  The Min-Max Theorem states that it
does not matter which player moves first, as long as we allow mixed
strategies for the players.

We model the task of choosing the parameter $p$ as the following
zero-sum game: Player $\A$, the algorithm designer, picks some
probability $p$, and player $\C$, the circuit designer, picks a value
$f$.  The outcome is $s_{p,f}$, the savings of the algorithm.  The
algorithm designer tries to maximize the savings, and the circuit
designer tries to minimize it.  The wire distribution of a circuit is
a mixed strategy for the circuit designer.  A mixed strategy for the
algorithm designer $\A$ would be a distribution on the probabilities.

A direct approach for designing the algorithm would be to select the
parameter $p$ depending on the circuit so that we obtain large
savings.  Specifically, given the wire distribution of the circuit
$\mathcal{F}$, the algorithm designer picks a $p$ and and the outcome
$s_{p,\mathcal{F}}$ is a convex combination of the values $s_{p,f}$.
Using the Min-Max Theorem we turn this game around: The algorithm
designer picks a mixed strategy and the circuit designer responds with
a pure strategy $f$, a circuit where all bottom-level gates have
weighted fan-in $f$.  The following lemma shows that there is a good
strategy for the algorithm designer.

\begin{lemma}
  There is a distribution $\mathcal{D}$ on parameters $p$ such that
  for all $f$,
        \begin{equation*}
          \expt_{p \sim \mathcal{D}} [s_{p,f}] \geq \frac1{c^{O\left(c^2\right)}}
        \end{equation*}
\end{lemma}
\begin{proof}
  Let $\mathcal{D}$ be the following distribution on $p$: For $I =
  O\left(c^2\log(c)\right)$ with suitable constants, and $1 \leq i
  \leq I$, we set $p = 2^{-i}$ with probability $A\cdot2^{-(I-i+1)}$,
  where $A=\frac1{\sum_{i=1}^I2^{-(I-i+1)}}$ is the normalization
  factor.  We know that $1 \leq A \leq 2$.  The expectation of $p$ is
  $\expt[p]= AI2^{-I-1}$.

  Let $f$ be any pure strategy of the circuit designer and $J =
  \log(f)$.  The expected outcome of the game for these strategies is
  \begin{equation*}
    \expt_{p \sim \mathcal{D}} [s_{p,f}] = \sum_{i=1}^{I} 2^{-(I-i+1)} s_{2^{-i},2^J}.
  \end{equation*}

  To lower bound the expected outcome, we use a case analysis on the
  savings similar to the one in Section \ref{sectionProof}.  Let $t =
  \frac{1}{4c}$ as defined in the previous section.  Let $I' \leq I$
  be the largest value such that for $i \leq I'$, we have $2^{J-i}
  \geq t$ and for $I' < i \leq I$ we have $2^{J-i} < t$.

  Using the savings from Lemma \ref{lem:Savings}, we have
  $s_{2^{-i},2^J} = 2^{-i-1} - \frac{c}{2^J}
  \log\left(c2^{J-i+1}\right)$ for $2^{J-i} \geq t$ and
  $s_{2^{-i},2^J} = 2^{-i-2}$ otherwise. The expected
  savings is then
        \begin{align*}
          \expt_{p \sim \mathcal{D}} [s_{p,f}] &= \sum_{i=1}^{I} 2^{-(I-i+1)} s_{2^{-i},2^J} \\
          &= \sum_{i=1}^{I'} 2^{-(I-i+1)} \left(2^{-i-1} - \frac{c}{2^J} \log\left(c2^{J-i+3}\right) \right) + \sum_{i=I'+1}^I 2^{-(I-i+1)} 2^{-i-2} \\
          &\geq \sum_{i=1}^{I} 2^{-(I+3)} - \sum_{i=1}^{I'} 2^{-(I-i+1)} \frac{c}{2^J} \log\left(c2^{J-i+3}\right) \\
          &= \frac1{2^{I+1}} \left(\frac{I}4 - c\sum_{i=1}^{I'}
            2^{-(J-i)} \log\left(c2^{J-i+3} \right) \right)
        \end{align*}

	Let $j = \lceil(J - i)\rceil$. By the definition of $I'$ we
        have $j \geq \log(t) = -\log(c) -2$. Hence
        \begin{align*}
          \sum_{i=1}^{I'} 2^{-(J-i)} \log\left(c2^{J-i+3}\right) &\leq \sum_{j=\log(t)}^{\infty}  2^{-j} \left(j +\log\left(8c\right) \right) \\
          &\leq 8c \log\left(8c\right) +  \sum_{j=1}^{\infty}  j 2^{-j} +  \log\left(8c\right) \\
          &= O\left(c \log(c)\right)
        \end{align*}
        
        Hence for $I = O\left(c^2\log(c)\right)$ we get
        \begin{equation*}
          \expt_{p \sim \mathcal{D}} s_{p,f} = \frac1{c^{O\left(c^2\right)}}
        \end{equation*}
\end{proof}

We now conclude that for every $f$ there is a $p = 2^{-i}$ for $1 \leq
i \leq I$, such that $s_{p,f} \geq
\frac1{c^{O\left(c^2\right)}}$. Using that for every mixed strategy
for $f$, the savings is a convex combination of the savings for pure
strategies, we conclude the same for any strategy on $f$.

This gives us the final algorithm: Given a circuit $C$ with wire
distribution $\mathcal{F}$, evaluate $\expt_{f \sim \mathcal{F}}
[s_{p,f}]$ with $p= 2^{-i}$ for each $1 \leq i \leq I$ as above and
use the optimal $p$ for the random restriction.

The savings is tight in the sense that there is a mixed strategy on
$f$ such that the expected savings is at most $1/2^{\Omega(c)}$.
\begin{lemma}
  There is a wire distribution $\mathcal{F}$ such that for any $p$
  \begin{equation*}
    \expt_{f \sim \mathcal{F}} [s_{p,f}] \leq \frac1{2^{\Omega(c)}}
  \end{equation*}
\end{lemma}
\begin{proof}
  Let $p$ be the strategy of the algorithm designer and let
  $\mathcal{F}$ be the distribution such that for $1 \leq j \leq c$,
  $c_{2^j} = 1$ and $c_f = 0$ for any other $f$. By lemma
  \ref{lem:Convex} we have
	\begin{equation*}
          \expt_{f \sim \mathcal{F}} [s_{p,f}] = \sum_{j=1}^{c} \frac1c s_{p,2^j}
	\end{equation*}
	We argue that for large $c$ and $p \geq \frac1{2^c}$, the
        savings is negative.  Assume $p \geq \frac1{2^c}$. There is
        some $j^* \leq c$ such that for $f=2^{j^*}$, $1 \leq pf \leq
        2$. Using that for any $p$ and $f$, the savings $s_{p,f}$ is
        upper bounded by $\frac{p}2$ we get
	\begin{align*}
		\expt_{f \sim \mathcal{F}} [s_{p,f}] &= \sum_{j=1}^{c} \frac1c s_{p,2^j} \\
		&\leq \frac{p}2 - \frac1c s_{p,2^{j^*}} \\
		&= \frac{p}2 - \frac1c \frac{c}{2^{j^*}} \log\left(cp2^{j^*+3}\right) \\
		&\leq \frac{p}2 \left(1 - \left(\log\left(8c\right) + 1 \right) \right)
	\end{align*}
	For large $c$, the expectation is therefore negative. On the
        other hand, if $p \leq \frac1{2^c}$, then $\expt_{f \sim
          \mathcal{F}} [s_{p,f}] \leq \frac{1}{2^{c-1}}$.
\end{proof}

\section{Conclusion}
In this paper, we present the first nontrivial algorithm for deciding
the satisfiability of $cn$-wire threshold circuits of depth $2$. The
same result also applies to the special of case of $\zeroone$ Integer
Linear Programming with sparse constraints.  The algorithm improves
over exhaustive search by a factor $2^{sn}$ where $s= 1/c^{O(c^2)}$.

%We regard our results as a step towards the understanding the deeper
%connections between lower bounds and satisfiability
%algorithms for threhsold circuits.

Several straightforward open questions remain.  Can we further improve the
savings?  The savings in our
algorithm is exponentially small in $c$, while the best known savings
for $cn$-size $\AC^0$ circuits is only polylogarithmically small in
$c$ \cite{ImpagliazzoMatthewsPaturi_2012_soda}.  Can we decrease this
gap?  If not, can we explain it in terms of the expressive power of the
circuits?

Our algorithm handles only linear size threshold circuits
of depth two.  Can we obtain nontrivial satisfiability algorithms for
slightly more relaxed models?  For example, it would be very
interesting to extend the result to larger depth circuits. 
It would also be nice to generalize the algorithm to deal with depth two
threshold circuits with linearly many gates.

It would
also be interesting to relax the restriction on the number of
wires. Unfortunately, as discussed earlier, 
it is  not be possible to obtain a constant 
savings algorithm for depth two threshold circuits of superlinearly
many wires
under $\seth$.
%If $\seth$ is true, there is no
%$\cnfsat$ algorithm for formulas with a superlinear number of clauses
%\cite{CalabroImpagliazzoPaturi_2006_ccc}. 
%Our current algorithm only works for circuits with linearly many wires.
%However, it might be
%possible to get an algorithm that depends on the number of
%bottom-level gates (i.e. clauses in the case of CNFs), as opposed to
%number of wires (i.e. literals).

It would be extremely interesting to find  a subquadratic algorithm 
for the Vector Domination Problem for dimension $\omega(\log n)$, 
which would imply the refutation of $\seth$.

% Our algorithm follows the general framework described by \cite{Santhanam_2010_focs}:  convert
% the input circuit to a simpler representation of the same function, and then solve
% Satisfiability for that representation.  In \cite{Santhanam_2010_focs}, the representation
% was a decision tree.  \cite{ImpagliazzoMatthewsPaturi_2012_soda} used a more general representation, a partition of
% the Boolean cube into sub-cubes where the function was constant.  Here, we are
% implicitly using an even more general representation:  a partition of the cube into
% affine sub-spaces over the reals on which the function is constant.   Unlike
% earlier work, the satisfiability problem is not easy even for the simpler representation;
% instead of being polynomial time for each sub-space, we get only an improved exponential
% time.  This motivates two further questions:  What is the exact complexity of finding
% a Boolean point in a sub-space given as a system of linear equations?  We reduce
% this problem to that of finding a solution to a 2-CSP satisfying exactly a given number
% of constraints, but there may be more direct approaches.  Another question raised is
% to look at the size of such a partition as a measure of complexity for Boolean functions.
% Can we find explicit functions requiring almost exponentially many such sub-spaces for
% such a partition?  Can we use this to prove lower bounds for threshold circuits? 

Our algorithm is a ``Split and List'' algorithm
\cite{Williams_2005_tcs}, split the variable set into subsets and list
all assignments to the subsets. As such, it inherently takes
exponential space.  Can we reduce the space requirement to polynomial
space? 

\noindent
{\bf Acknowledgments:}
We thank Dominik Scheder for the fruitful discussions on the Vector
Domination Problem. We also thank Shachar Lovett and Ryan Williams
for pointing us to a simpler algorithm for systems of linear
equations.

\bibliographystyle{plain}
\bibliography{../TeX/bib/complexity.bib}
\end{document}